\def\notes{0}
\def\pagenumtitle{1}
\def\mystatus{1} %

\documentclass[12pt]{article}

\usepackage{header}

\begin{document}

\title{
\vspace{-1em}
Separating Oblivious and Adaptive \\ Differential Privacy under Continual Observation
}
\author{
Mark Bun$^\ast$ \and 
Marco Gaboardi$^\ast$ \and 
Connor Wagaman\thanks{Boston University, \texttt{\{mbun,gaboardi,wagaman\}@bu.edu.}}
}

\date{April 2, 2026}

\maketitle

\vspace{-2em}

\begin{abstract}
We resolve an open question of Jain, Raskhodnikova, Sivakumar, and Smith (ICML 2023; \cite{JainRSS23}) by exhibiting a problem separating differential privacy under continual observation in the oblivious and adaptive settings.
The \emph{continual observation} (a.k.a.\ \emph{continual release}) model  formalizes privacy for streaming algorithms, where data is received over time and output is released at each time step. 
In the oblivious setting, privacy need only hold for data streams fixed in advance; in the adaptive setting, privacy is required even for streams that can be chosen adaptively based on the streaming algorithm's output.

We describe the first explicit separation between the oblivious and adaptive settings. 
The problem showing this separation is based on the correlated vector queries problem of Bun, Steinke, and Ullman (SODA 2017; \cite{BunSU19}). Specifically, we present an $(\varepsilon,0)$-DP algorithm for the oblivious setting that remains accurate for exponentially many time steps in the dimension of the input. On the other hand, we show that every $(\varepsilon,\delta)$-DP adaptive algorithm fails to be accurate after releasing output for only a constant number of time steps.

\end{abstract}

{
\newpage

\hypersetup{linkcolor=black}
\tableofcontents

\newpage
}

\section{Introduction}
\label{sec:intro}

Differential privacy (DP) \cite{DworkMNS16} is the standard framework for guaranteeing individual privacy when releasing statistics about a sensitive dataset. DP was initially considered in the ``batch model,'' where a trusted server holds a static dataset and privately answers queries posed by a data analyst.
However, many real-world datasets change over time, and a well-developed line of work, beginning with \cite{ChanSS11,DworkNPR10}, has investigated privacy guarantees for evolving datasets.
In the standard formulation of this setting, at each time step a new individual's data arrives and an output is released, with accuracy evaluated on the prefix of the dataset seen so far. Privacy requires indistinguishability of the entire output sequence for every pair of input datasets that differ by one individual.

Ensuring privacy in this streaming setting (known as \emph{continual observation} or \emph{continual release}) is challenging, since the arrival of one person's data may affect the algorithm's output at all subsequent time steps. 
Additionally, privacy for this setting is most meaningful when it accounts for adaptively selected inputs, where privacy must hold against an adversary who can specify the input stream in response to previous outputs from the algorithm. (Initial formulations of the continual observation model focused on the oblivious setting, where the input stream is fixed in advance but revealed to the algorithm timestep-by-timestep.)

This adaptive setting was first investigated implicitly by \cite{ThakurtaS13} (later by \cite{LecuyerSVGH19}) and formally defined by \cite{JainRSS23}, and it is especially relevant to machine learning. Stochastic gradient descent relies on accurately estimating the cumulative sum of individual gradients, where---as in the continual observation setting---one observation arrives at each time step (i.e., one gradient is computed at each iteration) and output is released at each time step (i.e., the model is updated).
Moreover, privacy for this problem must hold against adaptive inputs even when the input dataset is static, since the points at which we compute gradients are chosen based on the gradients computed so far. 
These connections between adaptive continual observation and private learning \cite{ThakurtaS13,KairouzMSTTX21} have led to an explosion of work on the continual observation setting---see \cite{PillutlaUpadhyayEtAl25} for a survey of some results.

In their work defining the adaptive model, \cite{JainRSS23} present two problems showing strong separations between the batch model and the oblivious continual observation model (specifically, these problems require additive error that is larger by a factor of $\wt\Omega(T^{1/3})$, where $T$ is the number of time steps at which output is produced). 
They note,
``It is open to separate the two continual release models in the sense our work separates the batch model and the continual release model: by providing problems that require a large error blowup in the more demanding model.'' That is, they ask,
\begin{center}
    \emph{Is there a problem separating oblivious and adaptive \\ differential privacy under continual observation?} \cite{JainRSS23}
\end{center}
We answer this question in the affirmative.

\subsection{Our Results}
\label{sec:results}

We show a problem that separates oblivious and adaptive differential privacy under continual observation. Informally, we prove the following theorem.

\begin{theorem}[\cref{thm:obliv-online-acc,thm:adapt-online-err}, informal]
    There is a problem $\cP^{d,T}$ parametrized by $d,T\in\N$ with the following properties.
    \begin{enumerate}
        \item For all $\eps \in \bigl(0,\frac{3}{2}\bigr]$ there exists $T = 2^{\Omega(\eps^4 d)}$ such that, for all sufficiently large $d\in\N$, there is an algorithm $\cA$ that is $(\eps,0)$-DP under \textbf{oblivious} continual observation and accurately answers $\cP^{d,T}$ for $T$ time steps.
        \item There exists some $T = O(1)$ such that, for all sufficiently large $d\in\N$, there is no algorithm $\cB$ that is $(\frac{1}{5},\frac{1}{20})$-DP under \textbf{adaptive} continual observation and accurately answers $\cP^{d,T}$ for $T$ time steps.
    \end{enumerate}
\end{theorem}

The problem yielding this separation is inspired by the \emph{correlated vector queries} problem of \cite{BunSU19}. Rather than focusing on the setting where data arrives over time and the same query is asked at every time step, that paper focuses on a related but orthogonal setting where the dataset is fixed and a new query is asked at every time step. They show exponential separations between the offline, oblivious online, and adaptive online query settings. To show the latter separation, they use correlated vector queries, a search problem where, for a private dataset $x\in\{\pm1\}^d$, each query specifies a set $V\subseteq\{\pm1\}^d$ and asks for an output $y\in\{\pm1\}^d$ that is $\alpha$-correlated with $x$ but nearly uncorrelated with all vectors in $V$ (beyond correlation explained by $x$).
In the oblivious setting, a single randomized response vector suffices to answer exponentially many such queries simultaneously. In contrast, in the adaptive setting, each new query can constrain the next answer to be nearly orthogonal to all previous answers, forcing the mechanism to reveal fresh information about $x$ at every step. After sufficiently many adaptive queries, this accumulated information enables reconstruction of a highly correlated estimate of $x$, contradicting differential privacy.

Our approach is similar, but does not follow in a black-box way from their results. In \cite{BunSU19}, each time step allows the analyst to specify an entirely new query---namely, an arbitrary collection $V \subseteq \set{\pm1}^d$ of vectors that simultaneously constrain the next answer.

In contrast, our continual observation setting is structurally more restricted. The private dataset arrives in two stages. First, $b \in \set{\pm1}^d$ arrives during a setup phase, during which no output is produced (or, alternatively, during which arbitrary output is produced and any answer is considered accurate). Then, during the arrival phase, a total of $T$ vectors $v_1,\ldots,v_T \in \set{\pm1}^d$ arrive, with only \emph{one} new vector arriving per time step. At time $t$, the mechanism must output a vector $y^{(t)}$ that is $\alpha$-correlated with $b$ while remaining nearly orthogonal to the prefix $v_1,\ldots,v_t$. Thus, rather than answering a sequence of independently specified search queries (each of which may contain many new constraints), we repeatedly answer the \emph{same} correlation task as the constraint set grows incrementally over time.

This distinction has important consequences for proving the lower bound. In \cite{BunSU19}, producing many nearly independent $\alpha$-correlated vectors directly enables reconstruction of the dataset $x$, yielding a contradiction to differential privacy. In our setting, however, reconstructing an arbitrary estimate of $b$ is not by itself sufficient to violate privacy. Instead, the reduction must be tailored to the continual observation model and show that the adaptive adversary can use the evolving constraint sequence to recover a specific \emph{challenge bit} embedded in $b$ that distinguishes the neighboring inputs. Establishing this stronger form of reconstruction also requires additional structure beyond a black-box use of the lower bound in \cite{BunSU19}.

In the oblivious setting, the entire sequence $v_1,\ldots,v_T$ is fixed in advance, so releasing $b$ via randomized response satisfies all constraints simultaneously with high probability. In the adaptive setting, each $v_{t+1}$ depends on previous outputs, forcing the mechanism to generate progressively less correlated and more independent views of $b$. We show that even under the one-constraint-per-round restriction, this adaptive process enables recovery of the challenge bit after only $O(1)$ steps, yielding the desired separation.

\subsection{Related Work}

As already described, our work relies heavily on the definition and open question of \cite{JainRSS23}, and the work of \cite{BunSU19}.
Other separation results between various streaming settings also inspired the question we address, including the separation of (nonprivate) oblivious and adaptive streaming from \cite{KaplanMNS21} and the work of \cite{CohenLNSS24} on separating privacy under offline and oblivious continual observation, though we do not make use of their techniques in our results.

There is also a rich literature on privacy under continual observation that inspires and situates our work. The continual observation (a.k.a.\ continual release) model was first proposed by \cite{DworkNPR10,ChanSS11}. The connections between (adaptive) continual observation and private learning shown by \cite{ThakurtaS13,KairouzMSTTX21} have led to a significant line of theoretical and empirical work in this space, with a focus on computing sums under continual observation---see \cite{PillutlaUpadhyayEtAl25} for a survey.
Recent works have developed new adaptively private algorithms with improved error guarantees and computational costs \cite{DenisovMRST22,Choquette-ChooDPGST24,DvijothamMPST24}, and there is a line of work on improving constant factors on the additive error for computing sums \cite{AnderssonP23,FichtenbergerHU23,HenzingerU25,HenzingerKU25}. 
Other work has investigated tasks related to counting distinct elements \cite{JainKRSS23} and various generalizations thereof \cite{QiuY25,AnderssonJS26,AamandCS26}.
There has also been work on continual observation algorithms for graph data in the node-private \cite{JainSW24} and  edge-private \cite{EpastoLMZ25,RaskhodnikovaS25,Zhou26} settings.

\subsection{Open Questions}

A possible direction for future work is to ask, \emph{Is there a ``more natural'' problem separating DP under oblivious and adaptive continual observation?}

Future work could also look for a further separation of the models: \emph{Is there some problem $\cP$ such that every algorithm $\cA$ that is $(\eps,\del)$-DP under oblivious continual observation and solves $\cP$ accurately is blatantly nonprivate under adaptive continual observation?}

The work of \cite{DenisovMRST22} makes some progress on this question. In particular, they show that any algorithm that is $(\eps,0)$-DP under oblivious continual observation remains $(\eps,0)$-DP under adaptive continual observation (though, as we show in \cref{thm:obliv-online-acc,thm:adapt-online-err}, accuracy may fail). Thus, any positive answer to this question must use a problem that requires an approximate-DP algorithm for accuracy.
The work of \cite{DenisovMRST22} also answers a related but orthogonal question, showing that some algorithms can be private under oblivious continual observation yet blatantly nonprivate under adaptive continual observation. As a simplified example of their demonstration, consider an algorithm that holds a sensitive dataset, releases a random $r\in\zo^\lambda$ at time $t=1$, and then outputs $0$ thereafter unless the incoming point is equal to $r$, in which case it reveals the dataset. This is $\paren{0,\tfrac{T}{2^\lambda}}$-DP under oblivious continual observation, but in the adaptive setting the adversary can set the second stream element to $r$ and force blatantly nonprivate behavior at $t=2$.

\section{Preliminaries}

We begin with some background on differential privacy, in both the non-streaming (``batch'') model and the streaming (``continual observation'') settings. Two datasets $x, x' \in \cX^n$ are \emph{neighbors} if they differ in the data of a single individual.

\begin{definition}[$(\eps,\delta)$-indistinguishability]
Let $\eps > 0$ and $\delta \in [0, 1]$.
Two random variables $R_1,R_2$ over outcome space $\cY$ are \emph{$(\eps,\delta)$-indistinguishable} (denoted $R_1\simed R_2$) if, for all $Y\subseteq \cY$, we have
\[
    \pr{R_1\in Y} \leq e^{\eps} \pr{R_2 \in Y} + \delta
    \quad \text{and} \quad
    \pr{R_2\in Y} \leq e^{\eps} \pr{R_1 \in Y} + \delta.
\]
\end{definition}

\begin{definition}[Differential privacy (DP) \cite{DworkMNS16}]
\label{def:dp batch}
Let $\eps > 0$ and $\delta \in [0, 1]$.
A randomized algorithm $\cM:\cU^* \to \cY$ is \emph{$(\eps,\del)$-DP} if for all pairs of neighboring inputs $x, x' \in \cU^*$, the distributions $\cM(x)$ and $\cM(x')$  are $(\eps,\del)$-indistinguishable: $\cM(x)\simed\cM(x').$
\end{definition}

\subsection{Differential Privacy under Continual Observation}

Differential privacy was generalized to the streaming setting by \cite{ChanSS11,DworkNPR10} and requires that the entire sequence of outputs remain indistinguishable up to a change in one individual's data in the input stream.
The oblivious setting, where an input stream is fixed up front but is made visible to the algorithm timestep-by-timestep, was first defined by \cite{DworkNPR10,ChanSS11}.
The adaptive setting, where each input is specified immediately prior to being revealed to the algorithm (and can thus be chosen in response to output from the algorithm), was first considered implicitly by \cite{ThakurtaS13} and formally defined by \cite{JainRSS23}.

\begin{definition}[DP under \textbf{oblivious} continual observation \cite{DworkNPR10,ChanSS11}]
\label{def:crt-obliv}
    Let $\mech$ be an algorithm that, given a data stream $\xvec = (x_1, \ldots, x_T)$, produces an output stream $\avec = (a_1,\ldots, a_T)$.
    The algorithm $\mech$ is \emph{$(\eps,\del)$-DP under \textbf{oblivious} continual observation} if, for all neighboring streams $\xvec,\xvec'\in \cX^T$ (i.e., $x_t\neq x'_t$ for at most one $t\in[T]$), we have
    \(
        \mech(\xvec) \simed \mech(\xvec').
    \)
\end{definition}

\begin{definition}[DP under \textbf{adaptive} continual observation \cite{JainRSS23}]
\label{def:crt-adapt}
    The \emph{view of $\adv$} in the privacy game $\Pi_{\mech,\adv}$ (\cref{alg:privacy-game}) consists of $\adv$'s internal randomness and the transcript of messages it sends and receives. Let \emph{$V_{\mech,\adv}^{(\side)}$} denote $\adv$'s  view  at the end of the game run with input $\side \in \{L,R\}$.
    
    An algorithm $\mech$ is \emph{$(\eps,\del)$-DP under \textbf{adaptive} continual observation} if, for all adversaries~$\adv$, we have
    \(
        V_{\mech,\adv}^{(L)}
        \simed
        V_{\mech,\adv}^{(R)}.
    \)
\end{definition}

\begin{algorithm}[ht]
\caption{Privacy game $\Pi_{\mech,\adv}$ for the adaptive continual release model}
\label{alg:privacy-game}
    \begin{algorithmic}[1]
        \Statex \textbf{Input:}  time horizon $T \in \N$, $\side \in \{L, R\}$ (not known to $\adv$).
        \For{\text{$t=1$ to $T$}}
        \State $\adv$ outputs $\rdtype_t \in \{\chall, \reg\}$; $\chall$ is chosen once during the game.
            \If{$\rdtype_t = \reg$}
                \State $\adv$ outputs $x_t \in \cX$ which is sent to $\mech$.
            \EndIf
            \If{$\rdtype_t = \chall$}
                \State $t^* \gets t$. 
                \State $\adv$ outputs $(x_t^{(L)},x_t^{(R)}) \in \cX^2$.
                \State $x_t^{(\side)}$ is sent to $\mech$.  
            \EndIf
            \State $\mech$ outputs $a_t$ which is given to $\adv$.
        \EndFor
    \end{algorithmic}
\end{algorithm}

\cite{CohenLNSS24} defines an offline version of \cref{def:crt-obliv}, in which an algorithm receives the entire input stream $\xvec$ at once, rather than item-by-item, and then outputs a full sequence $(a_1,\ldots,a_T)$. We say an algorithm for this setting is \emph{$(\eps,\del)$-DP under \textbf{offline} continual observation}.

\section{Separating Oblivious and Adaptive Privacy}
\label{sec:obliv-adaptive}

We show there is a problem that is ``easier'' to solve in the oblivious continual observation setting than in the adaptive continual observation setting---that is, for a specified privacy guarantee, we can give accurate answers at far more time steps in the former setting than in the latter setting. The problem yielding this separation is inspired by the problem of \emph{correlated vector queries} described by \cite{BunSU19} and defines accuracy with a very similar loss function. However, whereas we work with a data stream and ask the same query at every time step, \cite{BunSU19} works with a static dataset and asks a different query at every time step (see \cref{sec:results} for a further discussion on differences between our settings).

In \cref{def:sep-adapt-nonadapt} we introduce the problem and accuracy definition yielding our separation result. We next describe an algorithm for the oblivious setting that releases accurate answers at many time steps (\cref{sec:obliv-alg}), and then give an upper bound on the number of time steps at which accurate answers can be released in the adaptive setting (\cref{sec:adapt-bd}).

\begin{definition}[Problem separating online and online adaptive models]
\label{def:sep-adapt-nonadapt}
    Let $\alpha\in(0,1)$, $d\in\N$, and $T\in\N$ be parameters; we denote by $\cP^{\alpha,d,T}$ an instance of the separation problem with these parameters.
    Sensitive data arrives in two phases: a \emph{setup phase} and an \emph{arrival phase}. Outputs are produced only during the arrival phase.
    \begin{itemize}
        \item \textbf{Setup phase:} $d$ ``one-bit'' individuals $b_1,\ldots,b_d\in\set{\pm1}$ arrive; no outputs are returned.
        \item \textbf{Arrival phase:} $T$ ``vector'' individuals $v_1,\ldots,v_T\in\set{\pm1}^d$ arrive. After each arrival $t\in[T]$, the algorithm outputs $\yout\in\set{\pm1}^d$.
    \end{itemize}
    Let $b=(b_1,\ldots,b_d)$ and $v_{[t]}=(v_1,\ldots,v_t)$. The (boolean-valued) loss at time $t$ is
    \[
        L_t(b,\yout) \; \Longleftrightarrow \;
        \bigl|\langle \yout-\alpha b,b\rangle\bigr| \le \tfrac{\alpha^2 d}{100}
        \;\wedge\;
        \forall v\in v_{[t]},~ \bigl|\langle \yout-\alpha b,v\rangle\bigr| \le \tfrac{\alpha^2 d}{100}.
    \]
    An algorithm $\cM$ is \emph{$\beta$-accurate for $\cP^{\alpha,d,T}$} if, with probability at least $1-\beta$, the entire sequence of outputs $\yout[1],\ldots,\yout[T]$ satisfies the loss function.
\end{definition}

Informally, the loss function says that the output vector $y$ should be close to the vector $\alpha b$ (e.g., for small $\alpha$, the vector $y$ must have a small but precise correlation with $b$), and have as little correlation as possible with each $v\in v_{[t]}$ (roughly, some $v\in v_{[t]}$ may themselves be close to $\alpha b$, so correlation with these vectors is permitted, but the answer should otherwise be nearly orthogonal to all $v\in v_{[t]}$). 

\subsection{Algorithm for Oblivious Continual Observation}
\label{sec:obliv-alg}

We show there is a DP algorithm for this problem in the oblivious setting that can run accurately for exponentially many time steps (in the dimension of the vector individuals).

The algorithm $\cM$ for this setting is inspired by the algorithm in the proof of \cite[Theorem 4.1]{BunSU19}. At a high level, $\cM$ runs a randomized response-like algorithm \cite{Warner65} once on each one-bit individual, stores the resulting vector $y$, and returns this same vector $y$ at every time step. Privacy roughly follows from the privacy of randomized response, and accuracy follows from Hoeffding's inequality.

\begin{theorem}[Accuracy for oblivious continual observation]
\label{thm:obliv-online-acc}
    For every $0 < \alpha < 1/2$, there exists some $T = 2^{\Omega(\alpha^4 d)}$ such that, for every sufficiently large $d\in \N$, there is an algorithm $\cM$ that is $(1/T)$-accurate for $\cP^{\alpha, d, T}$ and is $(3\alpha,0)$-DP under \textbf{oblivious} continual observation.
\end{theorem}

\begin{proof}[Proof of \cref{thm:obliv-online-acc}]    
    We first define the algorithm $\cM$ for this setting. Let $\yuniv = (y_1,\cdots,y_d)$ be the random vector obtained by independently setting, for each $i\in[d]$,
    \[
        y_i =
        \begin{cases}
            + b_i & \text{w.p.\ $\frac{1+\alpha}{2}$} \\ 
            - b_i & \text{w.p.\ $\frac{1-\alpha}{2}$}.
        \end{cases}
    \]
    Define $\cM$ as the algorithm that releases this same vector $\yuniv$ at every time step $t\in [T]$.

    We now analyze the privacy of $\cM$. Since the output of $\cM$ depends only on $(b_1,\ldots, b_d)$, it suffices to consider two datasets $b$ and $b'$ that differ on some index $i$. Since each coordinate of $\yuniv$ is generated independently, it suffices to bound the likelihood ratio of outputs for coordinate $i$ as
    \[
        \frac{\Pr[\cM(b_i) = y_i ]}{\Pr[\cM(b'_i) = y_i ] }
        \leq \frac{(1+\alpha)/2}{(1-\alpha)/2}
        = \frac{1+\alpha}{1-\alpha}.
    \]
    We use the fact that, for $x\in[0,1)$ we have $\ln(1+x)\leq x$ and $\ln(1-x) \geq -\frac{x}{1-x}$. This gives $\ln\paren{\frac{1+\alpha}{1 - \alpha}} \leq 3\alpha$. A symmetric argument shows that the natural log of the likelihood ratio is bounded below by $-3\alpha$, so the algorithm is $(3\alpha, 0)$-DP.

    We next analyze the algorithm's accuracy. This analysis follows almost verbatim from the accuracy analysis of \cite[Theorem 4.1]{BunSU19}, which we repeat below for completeness.

    To prove accuracy, observe that since the output $y$ does not depend on the ``vector'' individuals, we can analyze $\cM$ as if the ``vector'' individuals were fixed and given all at once.
    First, observe that $\Ex[\yuniv] = \alpha b$.
    Thus we have
    \[
        \Ex_{\yuniv}[\langle \yuniv - \alpha b, b \rangle] = 0
        \quad\text{and}\quad
        \forall t\in[T],\; \Ex_{\yuniv}[\langle \yuniv - \alpha b, v_t\rangle] = 0.
    \]
    Since $b$ and every vector $v_1,\ldots, v_T$ is fixed independently of $\yuniv$, and the coordinates of $y$ are computed using independent randomness, both $\langle \yuniv , b \rangle$ and $\langle \yuniv , v_t \rangle$ for all $t\in[T]$ are the sum of $d$ independent $\pmo$-valued random variables. 
    Thus, we can apply Hoeffding's inequality\footnote{We use the following statement of Hoeffding's inequality: if $Z_1,\dots,Z_n$ are independent $\pmo$-valued random variables, and ${Z} = \sum_{i=1}^{n} Z_i$, then
    $\Pr\bracks{\abs{Z - \Ex[Z]} > C\sqrt{n}} \leq 2e^{-C^2/2}$.
    }
    and a union bound to conclude there is some absolute constant $c > 0$ such that
    \begin{align*}
        &\Pr_{\yuniv}\bracks{|\langle \yuniv - \alpha b, b \rangle| > \frac{\alpha^2 d}{100}} \leq 2\exp\paren{-c\alpha^4 d} \text{, and}\\
        &\Pr_\yuniv\bracks{\exists t \in [T]\textrm{ s.t. }|\langle \yuniv - \alpha b, v_t \rangle| > \frac{\alpha^2 d}{100}} \leq 2T\exp\paren{-c\alpha^4 d}.
    \end{align*}
    The theorem now follows by setting an appropriate choice of $T = 2^{\Omega(\alpha^4 d)}$ such that $2(T+1)\cdot\exp\left({-c\alpha^4 d}\right) \leq 1/T$.
    Thus $\cM$ is $(1/T)$-accurate for $\cP^{\alpha,d,T}$, completing the proof.
\end{proof}

\subsection{Error Lower Bound for Adaptive Continual Observation}
\label{sec:adapt-bd}

We now show that no DP algorithm for the adaptive setting can run accurately for even some constant (in the dimension of the input) number of time steps.

\begin{theorem}[Error required for adaptive continual observation]
\label{thm:adapt-online-err}
    For every $0 < \alpha < 1/2$, there exists some $T = O(1/\alpha^2)$ such that for all sufficiently large $d \in \N$, there is no algorithm $\cM$ that is $(1/100)$-accurate for the problem $\cP^{\alpha,d,T}$ and is $\paren{\frac{1}{5}, \frac{1}{20}}$-DP under \textbf{adaptive} continual observation.
\end{theorem}

Our proof is inspired by the structure of the attack used in proving \cite[Theorem 4.2]{BunSU19}. However, our result does not follow directly from \cite{BunSU19}, and moreover does not seem to follow directly from their attack.

At a high level, our attack uses the following strategy. Let $\yout[t]$ be the output produced after the arrival of vector individual $v_t$. Set the next vector-valued individual $v_{t+1} = \yout[t]$. 
Because $v_2,\ldots, v_t$ contains all outputs $y^{(1)},\ldots, \yout[t-1]$, every accurate algorithm will need to recompute a new output in response to seeing $v_{t+1} = \yout[t]$ as the next vector individual (otherwise the output will be too close to $v_{t+1}$ to satisfy the loss function). Intuitively, this recomputed query will leak information about the private dataset. We show that after $T$ queries for some  $T=O(1/\alpha^2)$, a ``reconstruction lemma'' \cite[Lemma 4.3]{BunSU19} means that the private dataset can be reconstructed with high probability, which violates privacy.

Our proof relies heavily on the following reconstruction lemma from \cite{BunSU19}.

\begin{lemma}[Reconstruction Lemma 4.3 from \cite{BunSU19}]
\label{lem:reconstruction}
    Fix parameters $p,q\in [0,1]$.
    Let $x \in \pmo^{d}$ and $\yout[1], \cdots, \yout[k] \in \pmo^d$ be vectors such that
    \begin{align*}
        &\forall 1 \leq j \leq k, ~~ \langle \yout[j] , x \rangle \geq p d, \quad \text{and} \\
        &\forall 1 \leq j < j' \leq k, ~~ |\langle \yout[j], \yout[j'] \rangle| \leq q d.
    \end{align*}
    If we let $\tilde{x} = \mathrm{sign}(\sum_{j=1}^{k} \yout[j]) \in \pmo^{d}$ be the coordinate-wise majority of $\yout[1],\ldots,\yout[k]$, then
    \[
        \langle \tilde{x}, x \rangle
        \geq
        \left(1-\frac{2}{p^2 k} - \frac{2(q-p^2)}{p^2} \right) d.
    \]
\end{lemma}

We also use \cref{fact:tv-dp} from the literature on differential privacy, which shows the following relationship between $(\eps,\del)$-indistinguishability and total variation distance. 

\begin{fact}[$\dtv$ and $(\eps,\del)$-indistinguishability]
\label{fact:tv-dp}
    Let $\eps > 0$ and $\delta \in [0, 1]$, and let $P$ and $Q$ be two probability distributions such that $P\simed Q$.
    Then
    \[
        \dtv\bigl( P , Q \bigr) \leq (e^\eps - 1) + \del.
    \]
    Moreover, if $\eps\leq 1$, then
    \(
        \dtv\bigl( \cA(x) , \cA(x') \bigr) \leq 2\eps + \del.
    \)
\end{fact}

We now prove \cref{thm:adapt-online-err}. 

\begin{proof}[Proof of \cref{thm:adapt-online-err}]
    We show there exists some $T = O(1/\alpha^2)$ such that the output $\yout[1],\ldots, \yout[T]$ of any algorithm $\cM$ that is $(1/100)$-accurate for $\cP^{\alpha,d,T}$ can be used to identify the value of the challenge bit in the privacy game (\cref{alg:privacy-game}) with probability at least $\frac{3}{4}$. In other words, the total variation distance between output distributions on two inputs that differ on the challenge bit is at least $\frac{1}{2}$, which by \cref{fact:tv-dp} contradicts the fact that the views of the adversary must be $\paren{\frac{1}{5}, \frac{1}{20}}$-indistinguishable for two inputs that differ in the challenge bit. In this proof, assume for contradiction that $\cM$ is $\bparen{\frac{1}{5}, \frac{1}{20}}$-DP under adaptive continual observation and $(1/100)$-accurate.

    Define the following adversary $\adv$. The adversary fixes a uniformly random bit string $b\in \pmo^d$. The adversary chooses a uniformly random index $\chall \in[d]$ and, for that index, chooses a uniformly random pair $\pair\in \set{(-1,+1), (+1,-1)}$. At each one-bit step $i\in[d]$, the adversary outputs type $\reg$ and sends $b_i$ to the privacy game, with the exception of time step $\chall$, at which the adversary presents $\chall$ and sends $\pair$.
    After completing the one-bit phase (at one-bit step $d$), the adversary sends a uniformly random bit string $v_1 \in \pmo^d$ as the first vector individual, and receives output $\yout[1]$. At all subsequent time steps $t\in \set{2,\ldots, T}$, the adversary sends $v_t = \yout[t-1]$.

    We now show that no algorithm $\cM$ guaranteeing $(1/100)$-accuracy for $\cP^{\alpha, d, T}$ is $\paren{\frac{1}{5}, \frac{1}{20}}$-DP under adaptive continual observation.
    The subsequent steps follow from the proof of \cite[Theorem 4.3]{BunSU19}. 
    Let $p = 99\alpha/100$ and $q = 51 \alpha^2/50$.
    Since $\cM$ is assumed to be accurate for $T$ time steps, with probability $99/100$, we obtain vectors $\yout[1],\dots,\yout[T] \in \pmo^{d}$ such that
    \begin{align*}
        \forall t\in[T] \quad
        \angles{\yout , b} &\geq \angles{\alpha b, b } - \abs{\angles{\yout - \alpha b, b}} \\
         &\geq \alpha d - \frac{\alpha^2d}{100} \\
         &\geq p d, \quad \text{and also}
    \end{align*}
    \begin{align*}
        \forall 1 \leq t < t' \leq T \quad
        \abs{\angles{ \yout, \yout[t'] }}
        &\leq \abs{\angles{\alpha b, \yout }} +  \abs{\angles{\yout[t'] - \alpha b, \yout}} \\
        & \leq \alpha \abs{\angles{\yout, b }} + \frac{\alpha^2 d}{100} \\
        & \leq \alpha\left( \abs{\angles{\alpha \bbits, \bbits }} +  \abs{\angles{\yout - \alpha\bbits, \bbits }}\right) + \frac{\alpha^2 d}{100} \\
        & \leq \alpha^2 d + \frac{\alpha^3 d}{100} + \frac{\alpha^2 d}{100} \\
        &\leq \frac{51}{50}\alpha^2 d 
        = q d.
    \end{align*}
Thus, by Lemma~\ref{lem:reconstruction}, for  $T \geq 1 + 100/\alpha^2$, $k\geq \max\set{T,d}$, and $\tilde{\bbits} = \mathrm{sign}(\sum_{j=1}^{k} \yout[j])$, we have
\begin{align*}
    \langle \tilde{\bbits}, \bbits \rangle 
    &\geq  \left(1-\frac{2}{p^2 k} - \frac{2(q-p^2)}{p^2} \right) d\\
    &\geq \left(1 - \frac{2}{(99\alpha/100)^2 d} - \frac{2(51\alpha^2/50-(99\alpha/100)^2)}{(99\alpha/100)^2} \right)d \\
    &= \left(1 - \frac{2(100/99)^2}{100} - 2\left(\frac{(51/50)-(99/100)^2}{(99/100)^2}\right)\right)d \\
    &\geq 0.89d.
\end{align*}

Recall that the challenge bit's location and value were chosen by the adversary uniformly at random. Therefore, $\cM$ cannot distinguish the challenge location from non-challenge locations, so because at least a $0.89$ fraction of the one-bit individuals' values were reconstructed correctly, this means the probability that the challenge value was correctly reconstructed by the adversary is at least $0.89$. This implies that the TV distance between adversary views is greater than $\frac{1}{2}$, which by \cref{fact:tv-dp} contradicts the fact that the views of the adversary must be $\paren{\frac{1}{5}, \frac{1}{20}}$-indistinguishable.
\end{proof}

\newpage
\printbibliography[heading=bibintoc]

\end{document}